\documentclass{llncs}       

\usepackage{amssymb}
\usepackage{graphicx}
\usepackage{amsmath}

\begin{document}

\title{The Weighted Barycenter Drawing Recognition Problem}
\author{Peter Eades\inst{1} \and  Patrick Healy\inst{2} \and Nikola S. Nikolov\inst{2}}

\institute{
University of Sydney,
\email{peter.d.eades@gmail.com} \and University of Limerick \email{patrick.healy,nikola.nikolov@ul.ie}}

\maketitle


\begin{abstract}
We consider the question of whether a given graph drawing $\Gamma$ of a triconnected planar graph $G$ is a weighted barycenter drawing. We answer the question with an elegant arithmetic characterisation using the faces of $\Gamma$. This leads to positive answers when the graph is a Halin graph,
and to a polynomial time recognition algorithm when the graph is cubic.
\end{abstract}

\section{Introduction}\label{se:intro}

The \emph{barycenter algorithm} of Tutte~\cite{tutte60,tutte63} is one of the earliest and most elegant of all graph drawing methods. It takes as input a graph $G=(V,E)$, a subgraph $F_0 = (V_0, E_0) $ of $G$, and a position $\gamma_a$ for each $a \in V_0$. The algorithm simply places each vertex $v \in V-V_0$ at the barycenter of the positions of its neighbours. The algorithm can be seen as the grandfather of force-directed graph drawing algorithms, and can be implemented easily by solving a system of linear equations. If $G$ is a planar triconnected graph, $F_0$ is the outside face of $G$, and the positions $\gamma_a$ for $a \in V_0$ are chosen so that $F_0$ forms a convex polygon, then the drawing output by the barycenter algorithm is planar and each face is convex.

The barycenter algorithm can be generalised to planar graphs with positive edge weights, placing each vertex $i$ of $V-V_0$ at the weighted barycenter of the neighbours of $i$. This generalisation preserves the property that the output is planar and convex~\cite{DBLP:journals/cagd/Floater97a}. Further, weighted barycenter methods have been used in a variety of theoretical and practical contexts~\cite{DBLP:conf/gd/FraysseixM03,DBLP:journals/comgeo/VerdierePV03,DBLP:journals/corr/abs-0708-0964,Thomassen83}. Examples of weighted barycenter drawings
(the same graph with different weights)
are in Fig.~\ref{fi:example}.

In this paper we investigate the following question: given a straight-line planar drawing $\Gamma$ of a triconnected planar graph $G$, can we compute weights for the edges of $G$ so that $\Gamma$ is the weighted barycenter drawing of $G$? We answer the question with an elegant arithmetic characterisation, using the faces of $\Gamma$. This yields positive answers when the graph is a Halin graph,
and leads to a polynomial time algorithm when the graph is cubic.

Our motivation in examining this question partly lies in the elegance of the mathematics, but it was also posed to us by Veronika Irvine (see \cite{DBLP:conf/gd/BiedlI17,tesselace}), who needed the characterisation to to create and classify ``grounds''
for bobbin lace drawings; this paper is a first step in this direction.
Further, we note that our result relates to the problem of morphing from one planar graph drawing to another (see \cite{DBLP:conf/gd/Barrera-CruzHL14,FloaterGotsman}). Previous work has characterised drawings that arise from the Schnyder algorithm (see~\cite{DBLP:conf/wg/BonichonGHI10}) in this context.
Finally, we note that this paper is the first attempt to characterise drawings that are obtained from force-directed methods.
\begin{figure}
  \centering
  \includegraphics[width=0.6\columnwidth]{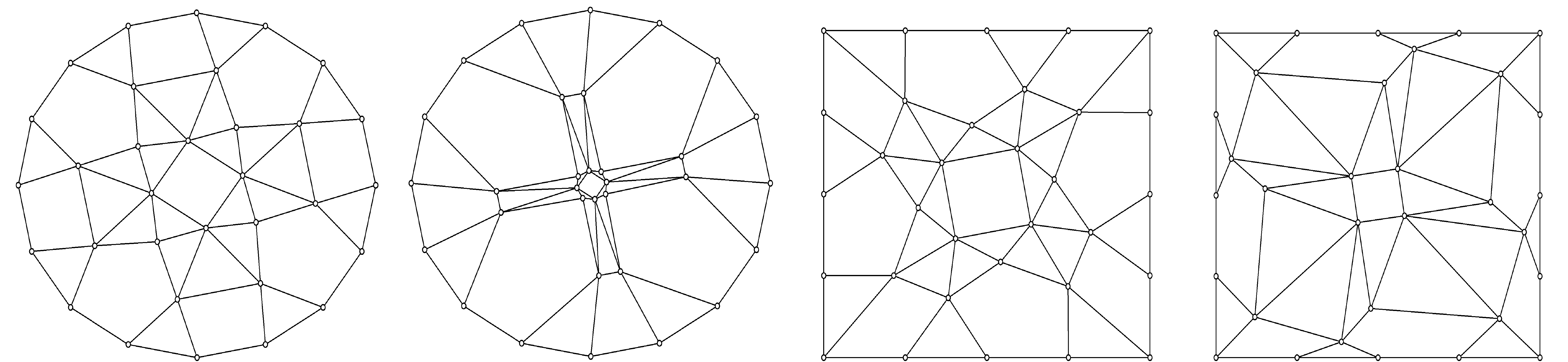}
  \caption{
Weighted barycenter drawings of the same graph embedding with different weights.
  }
  \label{fi:example}
\end{figure}

\section{Preliminaries: the weighted barycenter algorithm}\label{se:prelim}

Suppose that $G = (V,E)$ denotes a triconnected planar graph and $w$ is a 
\emph{weight function} that assigns a non-negative real weight $w_{ij}$ to each edge $(i,j) \in E$.
We assume that the weights are positive unless otherwise stated.
We denote $|V|$ by $n$ and $|E|$ by $m$.
In this paper we discuss planar straight-line drawings of such graphs;
such a drawing $\Gamma$ is specified by a position $\gamma_i$ for each vertex $i \in V$.
We say that $\Gamma$ is \emph{convex} if every face is a convex polygon.

Throughout this paper, $F_0$ denotes the outer face of a plane graph $G$.
Denote the number of vertices on $F_0$ by $f_0$.
In a convex drawing, the edges of $F_0$ form a simple convex polygon $P_0$.
Some terminology is convenient: we say that an edge or vertex on $F_0$ is \emph{external};
a vertex that is not external is \emph{internal};
a face $F$ (respectively edge, $e$) is \emph{internal} if $F$ (resp. $e$) is incident to an internal vertex,
and \emph{strictly internal} if every vertex incident to $F$ (resp. $e$) is internal.

The \emph{weighted barycenter algorithm} takes as input a triconnected planar graph $G = (V,E)$ with a weight function $w$,
together with $F_0$ and $P_0$, and produces a straight-line drawing $\Gamma$ of $G$ with $F_0$ drawn as $P_0$.
Specifically, it assigns a position $\gamma_i$ to each internal vertex $i$ such that $\gamma_i$ is the weighted barycenter of its neighbours in $G$.
That is:
\begin{equation}
\label{eq:bary0}
\gamma_i = \frac{1}{\sum_{j \in N(i)} w_{ij}} \sum_{j \in N(i)} w_{ij} \gamma_j
\end{equation}
for each internal vertex $i$. Here $N(i)$ denotes the set of neighbours of $i$.
If $\gamma_i = (x_i, y_i)$ then (\ref{eq:bary0}) consists of $2(n-f_0)$ linear equations in the $2(n-f_0)$ unknowns $x_i , y_i$.
The equations (\ref{eq:bary0}) are called the \emph{(weighted) barycenter equations} for $G$.
Noting that the matrix involved is a submatrix of the Laplacian of $G$, one can show that the equations have a unique solution
that can be found by traditional (see for example~\cite{Trefethen97})
or specialised (see for example~\cite{DBLP:journals/siamcomp/SpielmanT11}) methods.

The weighted barycenter algorithm, which can be viewed as a force directed method, was defined by Tutte~\cite{tutte60,tutte63} and extended by Floater~\cite{DBLP:journals/cagd/Floater97a};
the classic theorem says that the output is planar and convex:
\begin{theorem}
(Tutte~\cite{tutte60,tutte63}, Floater~\cite{DBLP:journals/cagd/Floater97a})
The drawing output by the weighted barycenter algorithm is planar, and each face is convex.
\label{th:tutte}
\end{theorem}



\section{The Weighted Barycenter Recognition Problem}
\label{se:TheProblemt}

This paper discusses the problem of finding weights $w_{ij}$ so that a given drawing is
the weighted barycenter drawing with these weights.
More precisely, we say that a drawing $\Gamma$ is a \emph{weighted barycenter drawing}
if there is a positive weight $w_{ij}$ for each internal edge $(i,j)$
such that for each internal vertex $i$, equations (\ref{eq:bary0}) hold.
\begin{description}
\item [{\bf The Weighted Barycenter Recognition problem}]
\item [{\em Input:}] A straight-line planar drawing $\Gamma$ of a triconnected plane graph $G = (V,E)$,
such that the vertices on the convex hull of $\{ \gamma_i : i \in V \}$ form a face of $G$.
\item [{\em Question:}] Is $\Gamma$ a weighted barycenter drawing?
\end{description}
Thus we are given the location $\gamma_i = (x_i , y_i)$ of each vertex,
and we must compute a positive weight $w_{ij}$ for each edge so that
the barycenter  equations (\ref{eq:bary0}) hold for each internal vertex.

Theorem~\ref{th:tutte} implies that if $\Gamma$ is a weighted barycenter drawing, then each face of the drawing is convex; however, the converse is false, even for triangulations (see Appendix).

\section{Linear Equations for the Weighted Barycenter Recognition problem}

In this section we show that the weighted barycenter recognition problem can be expressed in terms of linear equations.
The equations use \emph{asymmetric} weights $z_{ij}$ for each edge $(i,j)$; that is, $z_{ij}$ is not necessarily the same as $z_{ji}$.
To model this asymmetry we replace each undirected edge $(i,j)$ of $G$ with two directed edges $(i,j)$ and $(j,i)$; this gives a directed graph
$\overrightarrow{G} = (V,\overrightarrow{E})$.
For each vertex $i$, let $N^+(i)$ denote the set of \emph{out-neighbours} of $i$; that is,
$N^+(i) = \{j \in V : (i,j) \in \overrightarrow{E} \}.$

Since each face is convex, each internal vertex is inside the convex hull of its neighbours.
Thus each internal vertex position is a convex linear combination of the vertex positions of its neighbours.
That is, for each internal vertex $i$ there are non-negative weights $z_{ij}$ such that
\begin{equation}
\label{eq:bary3}
 \sum_{j \in N^+(i)}  z_{ij}   = 1 \text{~~~~and~~~~} \gamma_i = \sum_{j \in N^+(i)} z_{ij}  \gamma_j  .
\end{equation}
The values of  $z_{ij}$ satisfying (\ref{eq:bary3}) can be determined in linear time.
For a specific vertex $i$, the $z_{ij}$ for $j \in N^+(i)$ can be viewed as a kind of \emph{barycentric coordinates} for $i$.
In the case that $|N^+(i)| = 3$, these coordinates are unique.

Although equations (\ref{eq:bary0}) and (\ref{eq:bary3}) seem similar, they are not the same: one is directed, the other is undirected.
In general $z_{ij} \neq z_{ji}$ for directed edges $(i,j)$ and $(j,i)$, while the weights $w_{ij}$ satisfy $w_{ij} = w_{ji}$.
However we can choose a ``scale factor'' $s_i > 0$ for each vertex $i$,
and scale equations (\ref{eq:bary3}) by $s_i$. That is, for each internal vertex $i$,
\begin{equation}
\label{eq:baryScaled}
\gamma_i = \frac{1}{  \sum_{j \in N^+(i)} s_i  z_{ij} }  \sum_{j \in N^+(i)} s_i z_{ij}  \gamma_j  .
\end{equation}
The effect of this scaling is that we replace $z_{ij}$ by $s_i z_{ij}$ for each edge $(i,j)$.

We would like to choose a scale factor $s_i > 0$ for each internal vertex $i$ such that 
for each strictly internal edge $(i,j) \in E$,
$ s_i z_{ij} = s_j z_{ji}$;
that is, we want to find a real positive $s_i$ for each internal vertex $i$ such that
\begin{equation}
\label{eq:scale}
s_i z_{ij} - s_j z_{ji} = 0
\end{equation}
for each strictly internal edge $(i,j)$.

It can be shown easily that the existence of any nontrivial solution to (\ref{eq:scale}) implies the existence of a positive solution (see Appendix).

We note that any solution of (\ref{eq:scale}) for strictly internal edges gives weights $w_{ij}$ such that the barycenter equations (\ref{eq:bary0}) hold. We choose $w_{ij} = s_i z_{ij}$ for each (directed) edge $(i,j)$ that is incident to an internal vertex $i$. Equations (\ref{eq:scale}) ensure that $w_{ij} = w_{ji}$ for each strictly internal edge.
For edges which are internal but not strictly internal,
we can simply choose $w_{ij} = s_i z_{ij}$ for any value of $s_i$, since $z_{ji}$ is undefined.

Thus if equations (\ref{eq:scale}) have a nontrivial solution, then the drawing is a weighted barycenter drawing.

\subsubsection{The main theorem.}

We characterise the solutions of equations (\ref{eq:scale}) with an arithmetic condition on the faces of $\Gamma$.
This considers the product of the weights $z_{ij}$ around directed cycles in $G$: if the product around each strictly internal face in the clockwise direction is the same as the product in the counter-clockwise direction, then equations (\ref{eq:scale}) have a nontrivial solution.
\begin{theorem}
\label{th:cycleProduct}
Equations (\ref{eq:scale}) have a nontrivial solution if and only if for each strictly internal face $C = (v_0, v_1, \ldots , v_{k-1},  v_k = v_0)$ in $G$, we have
\begin{equation}
\label{eq:cycle}
\prod_{i=0}^{k-1} z_{v_i,v_{i+1}} = \prod_{i=1}^{k} z_{v_{i},v_{i-1}}.
\end{equation}
\end{theorem}
\begin{proof}

For convenience we denote $\frac{z_{ji}}{z_{ij}}$ by $\zeta_{ij}$ for each directed edge $(i,j)$; note that $\zeta_{ij} = 1/\zeta_{ji}$. Equations (\ref{eq:scale}) can be re-stated as
\begin{equation}
\label{eq:zeta}
s_i - \zeta_{ij} s_j = 0
\end{equation}
for each strictly internal edge $(i,j)$, and the equations (\ref{eq:cycle}) for cycle $C$ can be re-stated as 
\begin{equation}
\label{eq:cycleZeta}
\prod_{i=0}^{k-1} \zeta_{v_i,v_{i+1}} = 1.
\end{equation}
First suppose that equations (\ref{eq:zeta}) have nontrivial solutions $s_i$ for all internal vertices $i$,
and $C = (v_0, v_1, \ldots , v_{k-1},  v_k = v_0)$ is a strictly internal face in $G$.
Now applying (\ref{eq:zeta}) around $C$ clockwise beginning at $v_0$, we can have:
$$
s_{v_0} = \zeta_{v_0,v_1} s_{v_1} ~ = ~ \zeta_{v_0,v_1}  \zeta_{v_1,v_2}  s_{v_2} ~ = ~ \zeta_{v_0,v_1}  \zeta_{v_1,v_2}  \zeta_{v_2,v_3} s_{v_3} ~ = ~\ldots
$$
We can deduce that
\begin{equation*}
s_{v_0} = \left(   \prod_{i=0}^{j-1}   \zeta_{v_i,v_{i+1}}   \right)  s_{v_j} = \left(   \prod_{i=0}^{k-1}   \zeta_{v_i,v_{i+1}}   \right)  s_{v_k}
= \left(   \prod_{i=0}^{k-1}   \zeta_{v_i,v_{i+1}}   \right)  s_{v_0}
\end{equation*}
and this yields equation (\ref{eq:cycleZeta}).

Now suppose that equation 
(\ref{eq:cycleZeta}) holds for every strictly internal facial cycle of $G$.
We first show that equation (\ref{eq:cycleZeta}) holds for \emph{every} strictly internal cycle.
Suppose that (\ref{eq:cycleZeta}) holds for two cycles $C_1$ and $C_2$ that share a single edge, $(u,v)$,
and let $C_3$ be the sum of $C_1$ and $C_2$ (that is, $C_3 = ( C_1 \cup C_2 ) - \{ (u,v) \}$).
Now traversing $C_3$ in clockwise order gives the clockwise edges of $C_1$ (omitting $(u,v)$)
followed by the 
clockwise edges of $C_2$ (omitting $(v,u)$).
But from equation (\ref{eq:cycleZeta}),
the product of the edge weights $\zeta_{ij}$ in the clockwise order around $C_1$ is one,
and the product of the edge weights $\zeta_{i'j'}$ in the 
clockwise order around $C_2$ is one.
Thus the product of the edge weights $\zeta_{ij}$ in clockwise order around $C_3$ is $\frac{1}{\zeta_{uv} \zeta_{vu} } = 1$.
%
That is, (\ref{eq:cycleZeta}) holds for $C_3$.
Since the facial cycles form a cycle basis, it follows that (\ref{eq:cycleZeta}) holds for every cycle.

Now choose a reference vertex $r$, and consider a depth first search tree $T$ rooted at $r$.
Denote the set of directed edges on the directed path in $T$ from $i$ to $j$ by $E_{ij}$.
Let $s_{r} = 1$, and
for each internal vertex $i \neq r$, let 
\begin{equation}
\label{eq:s}
s_{i} = \prod_{(u,v) \in E_{ri}} \zeta_{uv}.
\end{equation}
Clearly equation~(\ref{eq:zeta}) holds for every edge of $T$.
Now consider a back-edge $(i,j)$ for $T$ (that is, a strictly internal edge of $G$ that is not in $T$),
and let $k$ denote the least common ancestor of $i$ and $j$ in $T$.
Then from (\ref{eq:s}) we can deduce that
\begin{equation}
\label{eq:moo1}
\frac{s_i}{s_j} = \frac{ \prod_{(u,v) \in E_{ri}} \zeta_{uv} }{ \prod_{(u',v')\in E_{rj}} \zeta_{u'v'} }
= \frac{ \prod_{(u,v) \in E_{ki}} \zeta_{uv} }{ \prod_{(u',v')\in E_{kj}} \zeta_{u'v'} }.
\end{equation}
Now let $C$ be the cycle in $\Gamma$ that consists of the reverse of the directed path in $T$ from $k$ to $j$,
followed by the directed path in $T$ from $k$ to $i$, followed by the edge $(i,j)$.
Since equation~(\ref{eq:cycleZeta}) holds for $C$, we have:
\begin{equation}
\label{eq:moo2}
1
= \left( \prod_{(v',u') \in E_{jk}} \zeta_{v'u'} \right) \left( \prod_{(u,v) \in E_{ki}} \zeta_{uv} \right) \zeta_{ij} 
= \left( \frac{  \prod_{(u,v) \in E_{ki}} \zeta_{uv} } { \prod_{(u'v') \in E_{kj}} {\zeta_{u'v'}} } \right)  \zeta_{ij} 
\end{equation}
Combining equations (\ref{eq:moo1}) and (\ref{eq:moo2}) we have $s_i = \zeta_{ij} s_j$
and so equation (\ref{eq:zeta}) holds for each back edge $(i,j)$.
We can conclude that (\ref{eq:zeta}) holds for all strictly internal edges.
\qed

\end{proof}

\section{Applications}

We list some implications of Theorem~\ref{th:cycleProduct} for cubic, Halin~\cite{DBLP:journals/jgaa/Eppstein16} and planar graphs with degree larger than three. Proofs of the corollaries below are straightforward.

\begin{corollary}
\label{cor:cubic1}
A drawing $\Gamma$ of a cubic graph is a weighted barycenter drawing if and only if equations (\ref{eq:scale}) have  rank smaller than $n-f_0$.
\qed
\end{corollary}

\begin{corollary}
\label{cor:cubic2}
For cubic graphs, there is a linear time algorithm for the weighted barycenter recognition problem.
\qed
\end{corollary}

For cubic graphs, the weights $z_{ij}$ are unique, and thus equations (\ref{eq:scale}) give a complete characterisation of weighted barycenter drawings. One can use Theorem~\ref{th:cycleProduct} to test whether a solution of equations (\ref{eq:scale}) exists,
checking equations~(\ref{eq:cycle}) in linear time.


\begin{corollary}
\label{cor:halin}
Suppose that $\Gamma$ is a convex drawing of a Halin graph such that the  internal edges form a tree.
Then $\Gamma$ is a weighted barycenter drawing.
\qed
\end{corollary}

\subsubsection{Graphs with degree larger than three.}

For a vertex $i$ of degree $d_i > 3$, solutions for equations (\ref{eq:bary3}) are not unique.  Nevertheless, these equations are linear, and we have 3 equations in $d_i$ variables.
Thus, for each vertex $i$, the solution $z_{ij}, j \in N(i)$, form a linear space of dimension at most $d_i-3$. In this general case, we have:
\begin{corollary}
\label{co:general}
A drawing $\Gamma$ of a graph $G$ is a weighted barycenter drawing if and only if there are solutions $z_{ij}$ to equations (\ref{eq:bary3}) such that the cycle equation (\ref{eq:cycle}) holds for every internal face.
\qed
\end{corollary}

Although Corollary~\ref{co:general} is quite elegant, it does not lead to an immediately practical algorithm because the equations (\ref{eq:cycle}) are not linear.

\section{Conclusion}\label{se:conclusions}

Force-directed algorithms are very common in practice, and drawings obtained from force-directed methods are instantly recognisable
to most researchers in Graph Drawing.
However, this paper represents the first attempt to give algorithms to recognise the output of a particular force-directed method,
namely the weighted barycenter method.
It would be interesting to know if the results of other force-directed methods can be automatically recognised.

\subsubsection*{Acknowledgements.}
We wish to thank Veronika Irvine for motivating discussions.

\bibliographystyle{abbrv}
\bibliography{baryCentreBib}

\pagebreak

\section*{Appendix}
\
\subsection*{A triangulation which is not a weighted barycenter drawing.}

The weighted barycenter algorithm can be viewed as a force directed method, as follows.
We define the \emph{energy} $\eta(i,j)$ of an internal edge $(i,j)$ by
\begin{equation}
\label{eq:energyEdge}
\eta(i,j) = \frac{1}{2} w_{ij}  \delta (\gamma_i,\gamma_j)^2 = \frac{1}{2} w_{ij}  \left(   (x_i - x_j)^2 + (y_i-y_j)^2    \right)\\
\end{equation}
where $\delta(,)$ is the Euclidean distance and $\gamma_i = (x_i , y_i)$. The energy $\eta(\Gamma)$ in the whole drawing is the sum of the internal edge energies.
Taking partial derivatives with respect to each variable $x_i$ and $y_i$ reveals that $\eta(\Gamma)$ is minimised precisely when
the barycenter equations (\ref{eq:bary0}) hold.

\begin{figure}
  \centering
  \includegraphics[width=0.4\columnwidth]{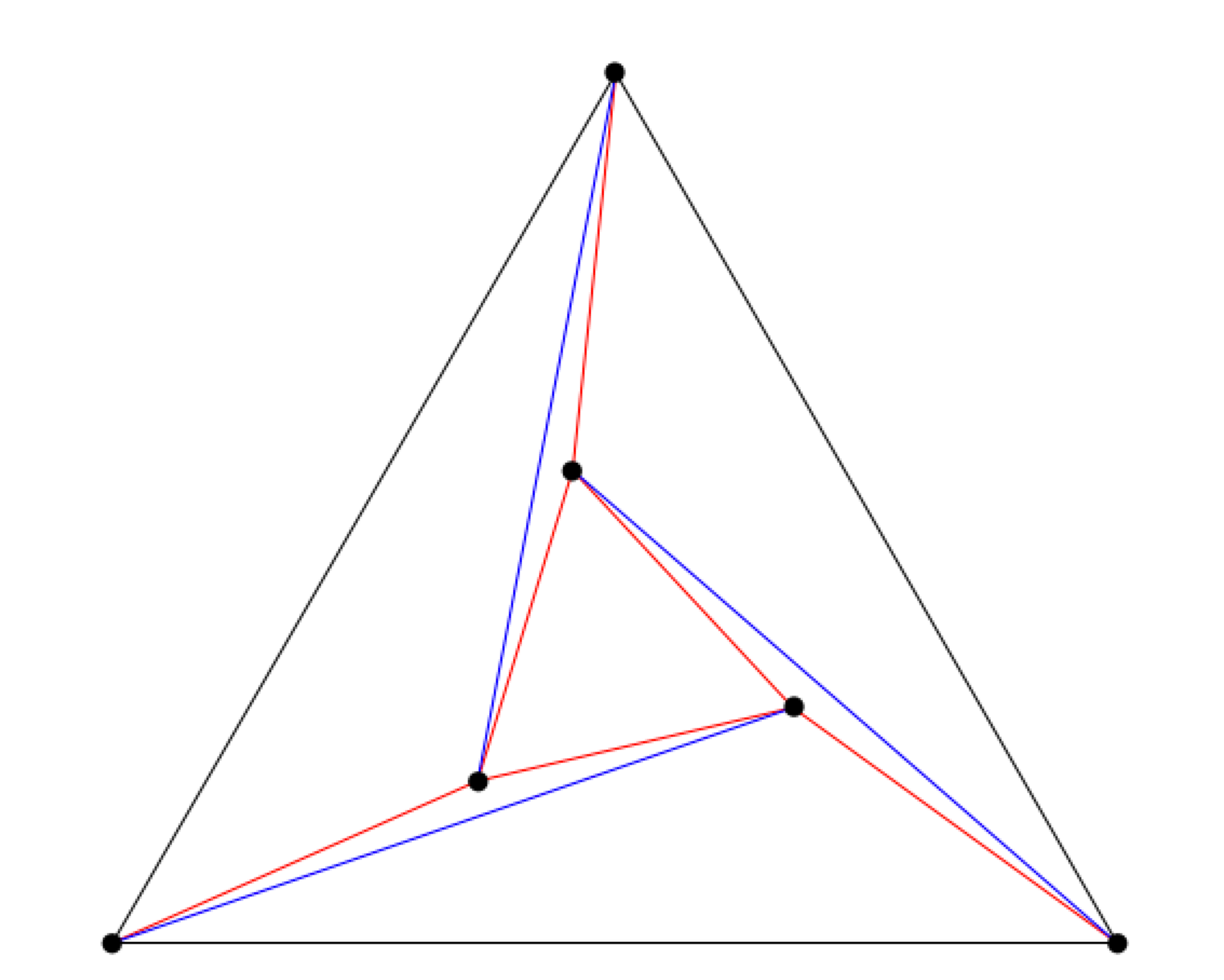}
  \caption{
 A triangulation which is not a weighted barycenter drawing.
  }
  \label{fi:triangles3}
\end{figure}
\begin{lemma}
The drawing in Fig.~\ref{fi:triangles3} is not a weighted barycenter drawing.
\end{lemma}
\begin{proof}
Suppose that the drawing $\Gamma$ in Fig.~\ref{fi:triangles3} is a weighted barycenter drawing with weights $w_{ij}$. 
The total energy $\eta(\Gamma)$ in the drawing is given by summing equation (\ref{eq:energyEdge}) over all internal edges,
and the drawing $\Gamma$ minimises $\eta(\Gamma)$.
Further, the minimum energy drawing is unique.
Consider the drawing $\Gamma'$ of this graph where the inner triangle is rotated clockwise by $\epsilon$, where $\epsilon$ is small.
The strictly internal edges remain the same length, while the edges between the inner and outer triangles become shorter.
Thus, since every $w_{ij} > 0$ for every such $(i,j)$,  $\eta(\Gamma') <\eta(\Gamma)$.
This contradicts the fact that $\eta$ is minimised at $\Gamma$.
\qed
\end{proof}

\subsection*{Positive solutions for equations (\ref{eq:scale}).}

\begin{lemma}
If equations (\ref{eq:scale}) have a nontrivial solution,
then they have a nontrivial solution in which every $s_i$ is positive.
\end{lemma}
\begin{proof}
Suppose that the vector $s$ is a solution to equations (\ref{eq:scale}),
and $s_i \leq 0$ for some internal vertex $i$.
Since $z_{ij} > 0$ for each $i,j$, and $G$ is connected, 
it is easy to deduce from (\ref{eq:scale}) that $s_j \leq 0$ for every internal vertex $j$.
Further if $s_i = 0$ for some internal vertex $i$ then $s_j = 0$ for every internal vertex.
Noting that $s$ is a solution to (\ref{eq:scale}) if and only if $-s$ is a solution, the Lemma follows.
\qed
\end{proof}

\subsection*{Properties of the coefficient matrix of the equations for the scale factors.}

Equations (\ref{eq:scale}) form a set of $m-f_0$ equations in the $n-f_0$ unknowns $s_i$.
We can write (\ref{eq:scale}) as
\begin{equation}
\label{eq:B}
B^T s=0
\end{equation}
where $s$ is an $(n-f_0) \times 1$ vector
and $B$ is an $(n-f_0) \times (m-f_0)$ matrix; precisely:
\[
B_{ie} = 
\begin{cases}
z_{ij} & \text{~if~} e \text{~is the edge~} (i,j) \in  \overrightarrow{E}\\
-z_{ij} & \text{~if~}e \text{~is the edge~} (j,i) \in  \overrightarrow{E}\\
0 & \text{otherwise.}
\end{cases}
\]
Note that $B$ is a weighted version of the directed incidence matrix of the graph $\overrightarrow{G}$.
Adapting a classical result for incidence matrices yields a lower bound on the rank of $B$:
\begin{lemma}
The rank of $B$ is at least $n - f_0 - 1$.
\end{lemma}
\begin{proof}
Since $G$ is triconnected, the induced subgraph of the internal vertices is connected.
Consider a submatrix of $B$ consisting of rows that correspond to the edges of a tree that spans the internal vertices.
It is easy to see that this $(n-f_0-1) \times (n-f_0) $ submatrix has full row rank (note that it has a column with precisely one nonzero entry).
The lemma follows.
\qed
\end{proof}
In fact, one can show (using the same method as in the proof of Theorem~\ref{th:cycleProduct}) that $B$ has rank exactly $n - f_0 - 1$
as long as the equations (\ref{eq:cycle}) hold for every strictly internal cycle.

\subsection*{Proofs of the corollaries.}

{\bf Corollary}~\ref{cor:cubic1}
\begin{proof}
In the case of a cubic graph, the weights $z_{ij}$ are unique. 
Thus, if the only solution to (\ref{eq:scale}) is $s_i = 0$ for every $i$, then the drawing is not a weighted barycenter drawing.
\qed
\end{proof}
%
{\bf Corollary}~\ref{cor:halin}
\begin{proof}

A \emph{Halin graph}~\cite{DBLP:journals/jgaa/Eppstein16} is a triconnected graph that consists of a tree, none of the vertices of which has exactly two neighbours, together with a cycle connecting the leaves of the tree. The cycle connects the leaves in an order so that the resulting graph is planar. A Halin graph is typically drawn so that the outer face is the cycle.

This corollary can be deduced immediately from Theorem~\ref{th:cycleProduct} since such a graph has no strictly internal cycles.
More directly, we can solve the equations (\ref{eq:scale}) starting by assigning $s_r = 1$ for the root $r$, and adding one edge at a time.
\qed
\end{proof}

\end{document}